\documentclass[12pt]{article}
\usepackage{amsmath,amsthm,amssymb}
\usepackage{mathrsfs}
\usepackage[T1]{fontenc}
\usepackage[latin9]{inputenc}
\usepackage{makeidx}
\usepackage{eufrak}
\DeclareMathOperator{\qdim}{qdim}
\DeclareMathOperator{\Span}{Span}

\begin{document}
\input amssym.def
\setcounter{equation}{0}
\newcommand{\wt}{{\rm wt}}
\newcommand{\spa}{\mbox{span}}
\newcommand{\Res}{\mbox{Res}}
\newcommand{\End}{\mbox{End}}
\newcommand{\Ind}{\mbox{Ind}}
\newcommand{\Hom}{\mbox{Hom}}
\newcommand{\Mod}{\mbox{Mod}}
\newcommand{\m}{\mbox{mod}\ }
\renewcommand{\theequation}{\thesection.\arabic{equation}}
\numberwithin{equation}{section}

\newtheorem{th1}{Theorem}
\newtheorem{ree}[th1]{Remark}
\newtheorem{thm}{Theorem}[section]
\newtheorem{prop}[thm]{Proposition}
\newtheorem{coro}[thm]{Corollary}
\newtheorem{lem}[thm]{Lemma}
\newtheorem{rem}[thm]{Remark}
\newtheorem{de}[thm]{Definition}
\newtheorem{hy}[thm]{Hypothesis}
\newtheorem{conj}[thm]{Conjecture}
\newtheorem{ex}[thm]{Example}
\newtheorem{theorem}{Theorem}[section]
\newtheorem{corollary}[theorem]{Corollary}
\newtheorem{lemma}[theorem]{Lemma}
\newtheorem{proposition}[theorem]{Proposition}
\theoremstyle{definition}
\newtheorem{definition}[theorem]{Definition}
\theoremstyle{remark}
\newtheorem{remark}[theorem]{Remark}
\newtheorem{example}[theorem]{Example}
\numberwithin{equation}{section}

\begin{center}
{\Large {\bf On \(S\)-matrix, and fusion rules for irreducible \(V^G\)-modules}}\\
\vspace{0.5cm}

Liuyi Zhang\\
Department of Mathematics, \\
University of
California, Santa Cruz, \\
CA 95064 USA \\
zhang.liuyi@qq.com\\
Li Wu\\
Chern Institute of Mathematics, \\
Key Lab of Pure Mathematics and Combinatorics
of Ministry of Education, \\
Nankai University, Tianjin 300071, People's Republic of China\\
nankai.wuli@gmail.com

\end{center}

\begin{abstract}
    \noindent Let \(V\) be a simple vertex operator algebra, and \(G\) a finite automorphism group of \(V\) such that \(V^G\) is regular. The definition of entries in \(S\)-matrix on \(V^G\) is discussed, and then is extended. The set of \(V^G\)-modules can be considered as a unitary space. In this paper, we obtain some connections between \(V\)-modules and \(V^G\)-modules over that unitary space. As an application, we determine the fusion rules for irreducible \(V^G\)-modules which occur as submodules of irreducible \(V\)-modules by the fusion rules for irreducible \(V\)-modules and by the structure of \(G\).
\end{abstract}

\section{Introduction}
\noindent This paper deals with fusion rules for irreducible \(V^G\)-modules which occur as submodules of irreducible \(V\)-modules. In \cite[Theorem 2]{tanabe2005intertwining}, a lower bound of those fusion rules is introduced. In the recent study of the orbifold theory, Dong-Ren-Xu \cite{2015arXiv150703306DOrbifold} prove that, assuming certain conditions, every irreducible \(V^G\)-module is a \(V^G\)-submodule of irreducible \(g\)-twisted \(V\)-modules for some \(g
\in G\). Motivated by the decomposition \cite[Equation 3.1]{2015arXiv150703306DOrbifold}, we show that the lower bound given by Tanabe is the desired fusion rules. \\
\\
\noindent Using the extended definition of entries in \(S\)-matrix, we establish a unitary space on the set of \(V^G\)-modules.
The inner product on the unitary space is derived from the entries \(S_{M_i,M_j}\). We show that the linear subspace spanned by \(g\)-twisted \(V\)-modules, \(\mathcal{S}_{V}(G)=\bigcup_{g\in G}\mathscr{M}(g)\), consists of exactly the annihilators of the linear subspace spanned by irreducible \(V^G\)-modules which occur as submodules of irreducible \(g\neq id\) twisted \(V\)-modules. \\
\\
\noindent Finding all desired entries \(S\)-matrix is a practical way to compute the fusion rules of irreducible modules of \(V^G\). Using the extended definition of entries in \(S\)-matrix, we show that the entries \(S_{M_{\lambda},E}\) are distributed proportionally to the quantum dimension of \(M_{\lambda}\).
The concept of quantum dimension in vertex operator algebra is first introduced in the paper \cite{dong2013quantum}.
As an application, we obtain the following main theorem.
\begin{theorem}
Let $M,N,F\in\mathscr{M}_V$. The notations are defined in Remark \ref{remarkInequality}.
Then, we have
\[
 N^{F_{\xi}}_{M_{\lambda_1},N_{\chi_1}}=
 \dim_{\mathbb{C}}\Hom_{\mathcal{A}_{\alpha_{\mathscr{O}(F)}}(G,\mathscr{O}(F))}(\mathrm{Ind}^{D(F)}_{S(F)}V_{\xi},\mathcal{I}_{\lambda_1,\chi_1}(v^1,v^2)).
\].
 \end{theorem}
\noindent The paper is organized as follows.
In Section 2, we review some results used in this paper for \(g\)-rational vertex operator algebras, the modular invariance of trace functions, Verlinde formula, and orbifold theory from \cite{dong1998twisted}, \cite{dong2000modular}, \cite{2015arXiv150703306DOrbifold}, \cite{huang2008vertex}.
In Section 3, we extend the definition of the entries in the \(S\) matrix, and then derive some properties based on the extended definition.
In Section 4, we establish a unitary space on the set of \(V^G\)-modules, and then analyze the structure of the unitary space.
This structure is useful in discussing the action of a group \(G\) on \(g\)-twisted modules. As an application, we show that \(G_M=G\) under certain assumptions.
In Section 5, we show the "evenly distributive property" of certain entries in the \(S\)-matrix, and then give a general formula of the fusion rules for irreducible \(V^G\)-modules which occur as submodules of irreducible \(V\)-modules by the fusion rules for irreducible \(V\)-modules and by the structure of \(G\).

\section{Preliminaries}
\subsection{Basics}
\noindent The definition of vertex operator algebra \(V=(V,Y,\mathbf{1},\omega)\) is introduced and is developed in \cite{borcherds1986vertex}, \cite{frenkel1989vertex}.
The definition of module (including weak, admissible and ordinary modules) are given in \cite{dong1997regularity}, \cite {dong1998twisted}. \\
\\
According to \cite{dong1998twisted}, a vertex operator algebra is called rational if any admissible module is a direct sum of irreducible admissible modules.
According to \cite{zhu1996modular}, a vertex operator algebra \(V\) is called \(C_2\)-cofinite if the subspace \(C_2(V)\) is spanned by \(u_{-2}v\) for all \(u,v\in V\) has finite codimension in \(V\).
A \(C_2\)-cofinite vertex operator algebra is finitely generated with a PBW-like spanning set \cite{Gaberdiel2003}. \\
\\
According to \cite{dong1997regularity}, a vertex operator algebra is called regular if any weak module is a direct sum of irreducible ordinary modules. In \cite{abe2004rationality}, \cite{LI1999495}, the authors show that regularity is equivalent to the combination of rationality and \(C_2\)-cofiniteness if \(V\) is of CFT type (\(V_0=\mathbb{C}\mathbf{1}\) and \(V_n=0\) for \(n<0\)). Regular vertex operator algebras include many important vertex operator algebras such as the lattice vertex operator algebras, vertex operator algebras associated to the integrable highest weight modules for affine Kac-Moody algebras, vertex operator algebras associated to the discrete series for the Virasoro algebra, the framed vertex operator algebras. \\
\\
Let \(V\) be a vertex operator algebra and \(g\) an automorphism of \(V\) of finite order \(T\). Then \(V\) is a direct sum of eigenspaces of \(g\):
\begin{equation*}
    V=\bigoplus_{r\in \mathbb{Z}/T\mathbb{Z}}V^r,
\end{equation*}
where \(V^r={v\in V| gv=e^{-2\pi ir/T}v}\).
Use \(r\) to denote both an integer between \(0\) and \(T-1\) and its residue class mod \(T\) in this situation.
The definitions of \(g\)-twisted \(V\)-module (including weak, admissible, ordinary modules) and \(g\)-rational vertex operator algebra are given in \cite[Definitions 2.1, 2.2, 2.3, 2.4]{2015arXiv150703306DOrbifold}.
\subsection{Modular Invariance}
\noindent Let \(V\) be a vertex operator algebra, \(g\) an automorphism of \(V\) of order \(T\) and \(M=\bigoplus_{n\in\frac{1}{T}\mathbb{Z}_+}M_{\lambda+n}\) a \(g\)-twisted \(V\)-module. \\
\\
For any homogeneous element \(v\in V\), the trace function associated to \(v\) is defined to be
\begin{equation*}
    Z_M(v,q)=\mathrm{tr}_Mo(v)q^{L(0)-\frac{c}{24}}=q^{\lambda-\frac{c}{24}}\Sigma_{n\in \frac{1}{T}\mathbb{Z}_+}\mathrm{tr}_{M_{\lambda+n}}o(v)q^n,
\end{equation*}
where \(o(v)=v(\mathrm{wt}v-1)\) is the degree zero operator of \(v\).
According to \cite{zhu1996modular}, \cite{dong2000modular}, the trace function \(Z_M(v,q\) converges to a holomorphic function on the domain \(|q|<1\) if \(V\) is \(C_2\)-cofinite.
Let \(\tau\) be in the complex upper half-plane \(\mathbb{H}\) and \(q=e^{2\pi i\tau}\). Then, the holomorphic function \(Z_M(v,q)\) becomes \(Z_M(v,\tau)\). \\
\\
Let \(v=1\) be the vacuum vector. Then, the holomorphic function \(Z_M(\mathbf{1},q)\) becomes the formal character of \(M\).
Denote \(Z_M(1,q)\) and \(Z_M(1,\tau)\) by \(\chi_M(\tau)\) and \(\chi_M(\tau)\), respectively.
We call \(\chi_M(q)\) the character of \(M\)\\
\\
In \cite{dong2000modular}, the authors discuss the action of \(\mathrm{Aut}(V)\) on twisted modules.
Let \(g,h\in \mathrm{Aut}(V)\) with \(g\) of finite order.
If \(M,Y_M\) is a weak \(g\)-twisted \(V\)-module, there is a weak \(h^{-1}gh\)-twisted \(V\)-module \(M\circ h,Y_{M\circ h}\), where \(M\circ h \cong M\) as vector spaces and
\begin{equation*}
    Y_{M\circ h}(v,z)=Y_M(hv,z),
\end{equation*}
for \(v\in V\).
This defines a left action of \(\mathrm{Aut}(V)\) on weak twisted \(V\)-modules and on isomorphism classes of weak twisted \(V\)-modules. Symbolically, write
\begin{equation*}
    (M,Y_M)\circ h= (M\circ h, Y_{M\circ h})=M\circ h,
\end{equation*}
where we sometimes abuse notation slightly by identifying \(M,Y_M\) with the isomorphism class that it defines. \\
\\
If \(g,h\) commute, \(h\) acts on the \(g\)-twisted modules.
Denote by \(\mathscr{M}_V(g)\) the equivalence classes of irreducible \(g\)-twisted \(V\)-modules and set \(\mathscr{M}_V(g,h)=\{M\in \mathscr{M}_V(g)|h \circ M \cong M\}\).
It is well known \cite{dong1998twisted}, \cite{dong2000modular} that if \(V\) is \(g\)-rational, both \(\mathscr{M}_V(g)\) and \(\mathscr{M}_V(g,h)\) are finite sets.
For any \(M\in\mathscr{M}_V(g,h)\), there is a \(g\)-twisted \(V\)-modules isomorphism
\begin{equation*}
    \varphi(h): M\circ h \rightarrow M.
\end{equation*}
The linear map \(\varphi(h)\) is unique up to a nonzero scalar. If \(h=1\), we simply take \(\varphi(1)=1\).
For \(v\in V\), set
\begin{equation*}
    Z_M(v,(g,h),\tau)=\mathrm{tr}_M o(v)\varphi(h)q^{L(0)-\frac{c}{24}}=q^{\lambda-\frac{c}{24}}\Sigma_{n\in \frac{1}{T}\mathbb{Z}_+}\mathrm{tr}_{M_{\lambda+n}}o(v)\varphi(h)q^n,
\end{equation*}
which is a holomorphic function on \(\mathbb{H}\) (see \cite{dong2000modular}).
Note that \(Z_M(v,(g,h),\tau)\) is defined up a nonzero scalar.
If \(h=1\), then \(Z_M(v,(g,1),\tau)=Z_M(v,\tau)\). \\
\\
The assumptions made in \cite{2015arXiv150703306DOrbifold} are
\begin{itemize}
    \item (V1) \(V=\bigoplus_{n\geq 0}V_n\) is a simple vertex operator algebra of CFT type,
    \item (V2) \(G\) is a finite automorphism group of \(V\) and \(V^G\) is a vertex operator algebra of CFT type,
    \item (V3) \(V^G\) is \(C_2\)-cofinite and rational,
    \item (V4) The conformal weight of any irreducible \(g\) twisted \(V\)-module for \(g\in G\) except \(V\) itself is positive.
\end{itemize}
In the rest of this paper, we also assume these four conditions. \\\
\\
The following results are obtained in \cite{abe2004rationality}, \cite{arXiv2014irreducible}, \cite{huang2015braided}.
\begin{lemma}
    Let \(V\) and \(G\) be as before. Then, \(V\) is \(C_2\)-cofinite, and \(V\) is \(g\)-rational for all \(g\in G\).
\end{lemma}
\noindent In Reference \cite{zhu1996modular}, Zhu introduced a second vertex operator algebra \((V,Y[\ ],\mathbf{1}, \tilde{\omega})\) associated to \(V\), where \(\tilde{\omega}=\omega-\frac{c}{24}\) and
\[
    Y[v,z]=Y(v,e^{z}-1)e^{z.\mathrm(wt)(v)}=\sum_{n\in \mathbb{Z}}v[n]z^{n-1}
\]
for homogeneous \(v\).
Write
\[
    Y[\tilde{\omega}, z]=\sum_{n\in \mathbb{Z}}L[n]z^{-n-2}.
\]
Carefully distinguish between the notion of conformal weight in the original vertex operator algebra and in the second vertex operator algebra \((V,Y[\ ],\mathbf{1}, \tilde{\omega})\). If \(v\in V\) is homogenous in the second vertex operator algebra, denote its weight by wt\([v]\). For such \(v\), define an action of the modular group \(\Gamma\) on \(T_M\) in a familiar way, namely
\begin{equation}
    \label{equationGeneralModularInvarianceFormula}
    Z_{M|\gamma}(v,(g,h),\tau)=(c\tau+d)^{-\mathrm{wt}[v]}Z_M(v,(g,h),\gamma\tau),
\end{equation}
where \(\gamma\tau\) is the Mobius transformation; that is,
\[
    \gamma: \tau \mapsto \frac{a\tau+b}{c\tau+d}, \ \gamma=\left(\begin{array}{cc}
a & b\\
c & d
\end{array}\right)\in \Gamma = SL\left(2,\mathbb{Z}\right).
\]
\noindent Let \(P(G)\) denote the commuting pairs of elements in a group \(G\). Let \(\gamma\in \Gamma\) act on the right of \(P(G)\) via
\[
    (g,h)\gamma=(g^ah^c,g^bh^d).
\]
The following results are from \cite{zhu1996modular}, \cite{dong2010characterization}.
\begin{theorem}\label{theoremModularInvariance}
    Assume \((g,h)\in P(\mathrm{Aut}(V))\) such that the orders of \(g\) and \(h\) are finite.
    Let \(\gamma=\left(\begin{array}{cc}
    a & b\\
    c & d
    \end{array}\right)\in \Gamma\). Also assume that \(V\) is \(g^ah^c\)-rational and \(C_2\)-cofinite.
    If \(M^i\) is an irreducible \(h\)-stable \(g\)-twisted \(V\)-module, then
    \[
        T_M|_{\gamma}(v,g,h,\tau)=\sum_{N_{j}\in \mathscr{M}(g^ah^c,g^bh^d)}\gamma_{i,j}(g,h)T_{N_{j}}(v,(g,h)\gamma,\tau),
    \]
    where \(\gamma_{i,j}(g,h)\) are some complex numbers independent of the choice of \(v\in V\).
    
    (2) The cardinalities $|\mathscr{M}(g,h)|$ and $|\mathscr{M}(g^a h^c,g^b h^d)|$ are equal for any $(g,h)\in P(G)$ and $\gamma\in \Gamma$.
 In particular, the number of inequivalent irreducible $g$-twisted $V$-modules  is exactly the number of  irreducible $V$-modules which are $g$-stable.
\end{theorem}
\noindent For convention, use \(\mathscr{M}_V\) for \(\mathscr{M}_V(1)\), the set of all irreducible \(V\)-modules.
\subsection{Fusion rules and Verlinde Formula}
\noindent Let \(V\) be as before and \(M,N,W\in \mathscr{M}_V\). The fusion rule \(N_{M,N}^W\) is \(\mathrm{dim}I_V
\begin{pmatrix}
    W\\
    \begin{matrix} M & N \end{matrix}
\end{pmatrix}\), where \(I_V\begin{pmatrix}
    W\\
    \begin{matrix} M & N \end{matrix}
\end{pmatrix}\) is the space of intertwining operators of type \(\begin{pmatrix}
    W\\
    \begin{matrix} M & N \end{matrix}
\end{pmatrix}\). \\
\\
Since \(V\) is rational, there is a tensor product \(\boxtimes\) of two \(V\)-modules (see \cite{huang1995theoryI}, \cite{huang1995theoryII}, \cite{huang1995theoryIII}) such that if \(M,N\) are irreducible then \(M\boxtimes N=\sum_{W\in \mathscr{M}_V}N_{M,N}^W W\). The irreducible \(V\)-module is called a simple current if \(M\boxtimes N\) is irreducible again for any irreducible module \(N\). \\
\\
The following Verlinde formula (see \cite{verlinde1988fusion}) is proved in \cite{huang2008vertex}.
\begin{theorem}
    \label{theoremVerlindeHuang}
    Let \(V\) be a rational and \(C_2\)-cofinite simple vertex operator algebra of CFT type and assume \(V\cong V'\). Let \(S=(S_{i,j})\) be the \(S\)-matrix. Then,
    \begin{itemize}
        \item (1) \((S^{-1})_{i,j}=S_{i,j'}=S_{i',j}\), and \(S_{i',j'}=S_{i,j}\).
        \item (2) \(S\) is symmetric and \(S^2=(\delta_{i,j'})\).
        \item (3) \(N_{i,j}^{k}=\sum_{s=0}^{d}\frac{1}{S_{0,s}}{S_{j,s}S_{i,s}S_{k',s}}\).
        \item (4) The \(S\)-matrix diagonalizes the fusion matrix \(N(i)=(N_{i,j}^{k})_{j,k=0}^{d}\) with diagonal entries \(\frac{S_{i,s}}{S_{0,s}}\), for \(i,s\in \{0,1,\cdots, d\}\). More explicitly, \(S^{-1}N(i)S=\mathrm{diag}(\frac{S_{i,s}}{S_{0,s}})_{s=0}^{d}\). In particular, \(S_{0,s}\neq 0\), for \(s\in \{0,1,\cdots, d\}\).
    \end{itemize}
\end{theorem}
We also have \cite{dong2015congruence}:
\begin{proposition}
    The \(S\)-matrix is unitary and \(S_{V,M}=S_{M,V}\) is positive for any irreducible \(V\)-module \(M\).
\end{proposition}
\subsection{On Orbifold Theory}
\noindent The quantum Galois theory was introduced in \cite{dong1997},\cite{hanaki1999quantum},\cite{dong2013quantum}.
\begin{theorem}
    (Quantum Galois Theory)Let \(V\) be a simple vertex operator algebra, \(G\) a compact subgroup of \(\mathrm{Aut}(V)\) acting continuously on \(V\). Then, as a \(G,V^G\)-module,
    \begin{equation}
        \label{equationQuantumGaloisTheory}
        V=\bigoplus_{\chi\in \mathrm{Irr}(G)}(W_{\chi}\otimes V_{\chi}),
    \end{equation}
    where
    \begin{itemize}
        \item \(V_{\chi}\neq 0\), \(\forall \chi \in \mathrm{Irr}(G)\),
        \item \(V_{\chi}\) is an irreducible \(V^G\)-module, \(\chi\in \mathrm{Irr}(G)\),
        \item \(V_{\chi}\cong V_{\lambda}\) as \(V^G\)-modules if and only if \(\chi=\lambda\).
    \end{itemize}
\end{theorem}
\noindent Some conjectures on orbifold theory are proved in \cite{2015arXiv150703306DOrbifold}.
\begin{theorem}
    \label{theoremGeneralizedQuantumGaloisTheory}
    (orbifold theory)Let \(V\), \(G\) be defined in the reference, and \(M\) be an irreducible \(g\)-twisted \(V\)-module, \(N\) an irreducible \(h\)-twisted \(V\)-module. Assume that \(M\) and \(N\) are not in the same orbit of \(\mathcal{S}\) under the action of \(G\). Then, as a \(\mathbb{C}^{\alpha_M}[G_M],V^{G_M}\)-module,
    \begin{equation}
        \label{equationGeneralizedQuantumGaloisTheory}
        M=\bigoplus_{\lambda\in \Lambda_{G_M,\alpha_M}}W_{\lambda}\otimes M_{\lambda},
    \end{equation}
    where
    \begin{itemize}
        \item \(W_{\lambda}\otimes M_{\lambda}\) is nonzero for any \(\lambda \in \Lambda_{G_M,\alpha_M}\),
        \item each \(M_{\lambda}\) is an irreducible \(V^{G_M}\)-module,
        \item \(M_{\lambda}\cong M_{\gamma}\), as \(V^{G_M}\)-modules, if and only if \(\lambda=\gamma\),
        \item each \(M_{\lambda}\) is an irreducible \(V^{G}\)-module,
        \item \(M_{\lambda}\ncong N_{\mu}\), as \(V^{G}\)-modules,
        \item any irreducible \(V^G\)-module is isomorphic to \(M_{\lambda}\) for some irreducible \(g\)-twisted \(V\)-module \(M\) and some \(\lambda\in \Lambda_{G_M,\alpha_M}\).
    \end{itemize}
\end{theorem}
\noindent The following results are from \cite{2015arXiv150703306DOrbifold}.
\begin{theorem}
    \label{Qd1}
    We have
    $$\mathrm{qdim}_{V^G}M=|G|\mathrm{qdim}_VM.$$
    for \(M\in \mathscr{M}_{V}(g)\), and \(g\in G\).
\end{theorem}
\begin{theorem} Use the same notations in Theorem \ref{theoremGeneralizedQuantumGaloisTheory}. We have
\begin{equation}
    S_{V_{\lambda},V^G}=\frac{\mathrm{dim}W_{\lambda}}{|G_M|}S_{M,V}
\end{equation}
\begin{equation}
    \label{equationOrbifoldProportional}
    \mathrm{qdim}_{V^G}M_{\lambda}=|G:G_M|(\mathrm{dim}W_{\lambda})(\mathrm{qdim}_{V}M).
\end{equation}
\end{theorem}
\begin{theorem}\label{Qd2} we have the following relation,
    \[
        \mathrm{glob}(V^G)=|G|^2\mathrm{glob}(V).
    \]
\end{theorem}

\newpage
\section{Definition extension of the entries in the \(S\)-matrix}
\noindent In this paper, \(V^G\) is a rational vertex operator algebra satisfying the four conditions (V1)--(V4) in Section 4, where \(G\) is a finite automorphism group of \(V\).
We first introduce some notations used in this paper.
\begin{remark}
    \label{remarkType}
    The last statement in Theorem \ref{theoremGeneralizedQuantumGaloisTheory} shows that there are two types of irreducible \(V^{G}\) modules modules.
    \begin{itemize}
        \item An irreducible \(V^{G}\) module \(M\) is of \emph{type one} if \(M\) occurs in the decomposition of irreducible \(V\) modules, as \(V^{G}\) modules.
        \item An irreducible \(V^{G}\) module \(M\) is of \emph{type two} if \(M\) does not occur in the decomposition of irreducible \(V\) modules, as \(V^{G}\) modules. That is, \(M\) occurs in a \(g\) twisted \(V^{G}\) module for some \(g\in G\) and \(g\neq 1\).
    \end{itemize}
\end{remark}
\begin{remark}
    \label{lemmaSMFi=0}
    Denote by \(\mathscr{M}_{V^G,\mathrm{I}}\) the set of irreducible \(V^G\)-modules of type one, by \(\mathscr{M}_{V^G,\mathrm{II}}\) the set of irreducible \(V^G\)-modules of type two. Set \(\mathcal{S}_{V}(G)=\bigcup_{g\in G}\mathscr{M}(g)\). Let \(\mathscr{O}(W)=\{M\in \mathscr{M}_V|M\circ g\cong W, \ \mathrm{for}\ \mathrm{some} \  g\in G\}\). Set \(\mathscr{S}_V(G)=\mathcal{S}_{V}(G)/G=\{\mathscr{O}(M)|M\in \mathcal{S}_{V}(G)\}\).
\end{remark}
\noindent The following definition of the \(S\)-matrix is well known (see \cite{zhu1996modular}).
\begin{definition}
    In Equation \ref{equationGeneralModularInvarianceFormula}, let \(g=h=1\), and \(\gamma=S=
    \begin{pmatrix}
        0&-1\\
        1&0
    \end{pmatrix}\). Then,
    \begin{equation}
        Z_{M^i}(v,-\frac{1}{\tau})=\tau^{\mathrm{Wt}[v]}\sum_{j=0}^{d}S_{i,j} Z_{M^j}(v,\tau).
    \end{equation}
    The matrix \(S=(S_{i,j})\) is called an \(S\)-matrix which is independent of the choice of \(v\).
\end{definition}
\noindent Then, we can extend the definition of the entries \(S_{U,W}\).
\begin{definition}
    \label{definitionExtension}
    Extend the notation of the entry\(S_{U,W}\), where \(U,W\in \mathscr{M}_V\).
    Add to \(S\) a superscript \(V\), where \(V\) is the associated vertex operator algebra.
    That is, \(S^{V}\) represents the \(S\)-matrix of \(V\), whereas \(S^{V^G}\) represents the \(S\)-matrix of \(V^G\).
    The entries \(S_{U,W}^{V}\), \(S_{U,W}^{V^G}\) are defined in the same way.
    Denote \(U=\bigoplus_{i=1}^{n}U_i\) and \(W=\bigoplus_{j=1}^{n}W_j\), where \(U_i\) and \(W_j\) are irreducible \(V^G\)-modules, for \(i \in \{1,2,\cdots,n\}\) and \(j \in \{1,2,\cdots,m\}\). Note that \(U_{i_1}\) and \(U_{i_2}\) might be isomorphic \(V^G\)-modules for \(i_1\neq i_2\).
    Then, define
    \begin{equation}
        S_{U,W}^{V^G}=\sum_{i=1}^{n}\sum_{j=1}^{m}S_{U_i,W_j}^{V^G}.
    \end{equation}
    This extension is well defined, because the module decomposition \(U=\bigoplus_{i=1}^{n}U_i\) is unique. The superscript \(V\) can be omitted, if no confusion follows.
\end{definition}
\noindent The following lemma shows that two \(V\)-modules can exactly be distinguished by the associated entries in the \(S\)-matrix.
\begin{lemma}
    \label{lemmaSameModule}
    Let \(U,W\) be two \(V\)-modules. Then, \(U\cong W\) if and only if \(S_{U,M}^V=S_{W,M}^V\) for every \(M\in \mathscr{M}_V\).
\end{lemma}
\begin{proof}
      Since \(V\) is a rational vertex operator algebra, \(U,W\) are completely reducible and the module decomposition is unique. For \(v\in V^G\) and \(W\in\mathscr{M}_{V^G}\), and , \(Z_W(v,\tau)\) are linearly independent (see \cite{zhu1996modular}).
      Thus, \(U\cong W\), if and only if \(Z_U(v,-\frac{1}{\tau})=Z_W(v,-\frac{1}{\tau})\), if only if  \(S_{U,M}^V=S_{W,M}^V\) for every \(M\in \mathscr{M}_V\).
\end{proof}
\noindent The formula in the next lemma is a variation of the third statement in Theorem \ref{theoremVerlindeHuang}.
\begin{lemma}
    \label{lemmaBasicS-Fusion}
    Let \(U\) and \(W\) be two \(V\)-modules. Let \(M^k\) be an irreducible \(V\)-module, where \(k\in \{0,1,\cdots,d\}\). Then, \(S_{U\boxtimes W, M^k}=\frac{1}{S_{V,M^k}}S_{U,M^k}S_{W,M^k}\).
\end{lemma}
\begin{proof}
    It is sufficient to show the equation is true when \(U\) and \(W\) are irreducible \(V\)-modules. Assume that \(U\) and \(W\) are irreducible \(V\)-modules. Let \(v\in V\) be a homogeneous vector in the second vertex operator algebra. Recall the fusion product relation, \(U\boxtimes W=\sum_{k=0}^{d}N_{U,M}^{M^k}M^k\).
    Then, by Verlinde formula,
    \begin{align*}
        Z_{U\boxtimes W}(v,-\frac{1}{\tau})=&Z_{\{\sum_{k=0}^{d}N_{U,M}^{M^k}M^\}}(v,-\frac{1}{\tau})\\
        =&\sum_{k=0}^{d}N_{U,M}^{M^k}Z_{M^k}(v,-\frac{1}{\tau})\\
        =&\sum_{k=0}^{d}(\sum_{l=0}^{d}\frac{1}{S_{V,M^l}}S_{U,M^l}S_{W,M^l} S_{(M^k)',M^l})Z_{M^k}(v,-\frac{1}{\tau})\\
        =&\sum_{k=0}^{d}(\sum_{l=0}^{d}\frac{1}{S_{V,M^l}}S_{U,M^l}S_{W,M^l} S_{\{(M^k)',M^l\}})(\tau^{\mathrm{Wt}[v]}\sum_{r=0}^{d}S_{M^k,M^r}Z_{M^r}(v,\tau))\\
        =&\tau^{\mathrm{Wt}[v]}(\sum_{l=0}^{d}\frac{1}{S_{V,M^l}}S_{U,M^l}S_{W,M^l} )(\sum_{r=0}^{d}\delta_{M^l,M^r}Z_{M^r}(v,\tau))\\
        =&\tau^{\mathrm{Wt}[v]}\sum_{l=0}^{d}\frac{1}{S_{V,M^l}}S_{U,M^l}S_{W,M^l}Z_{M^l}(v,\tau)
    \end{align*}
    On the other hand, \(Z_{U\boxtimes W}(v,-\frac{1}{\tau})=\tau^{\mathrm{Wt}[v]}\sum_{l=0}^{d}S_{U\boxtimes W,M^l}Z_{M^l}(v,\tau)\). A comparison of the coefficients of \(Z_{M^l}(v,\tau)\) yields \(S_{U\boxtimes W, M^l}=\frac{1}{S_{V,M^l}}S_{U,M^l}S_{W,M^l}\).
\end{proof}
\noindent Next, we need to find some methods to compute the entries \(S_{M,W_j}^{V^G}\) based on entries in the \(S\)-matrix of \(V\).
\begin{lemma}
    \label{lemmaOrbits}
    Let \(M,W\in\mathscr{M}_{V}\).
    By orbifold theory, denote \(W=\bigoplus_{j=1}^{n}W_{j}\otimes U_j\), where \(W_{j}\in \mathscr{M}_{V^G}\), and \(W_{j_1}\ncong W_{j_2}\), for \(j_1\neq j_2\).
    Then,
    \begin{equation}
        \label{equationSEntry}
        S_{M,W_j}^{V^G}=(\mathrm{dim}U_j)\sum_{N\in \mathscr{O}(W)}S_{M,N}^V.
    \end{equation}
\end{lemma}
\begin{proof}
    Denote \(M=\bigoplus_{i=1}^mM_i\), where \(M_i\in \mathscr{M}_{V^G}\), for \(i\in \{1,2,\cdots,m\}\). Note that \(M_{i_1}\) and \(M_{i_2}\) might be isomorphic, for \(i_1\neq i_2\).
    Definition \ref{definitionExtension} shows that
    \begin{align*}
        Z_M(v,-\frac{1}{\tau})=&\tau^{\mathrm{Wt}[v]}\sum_{N\in \mathscr{M}_V}S_{M,N}^{V}Z_N(v,\tau)\\
        =&\tau^{\mathrm{Wt}[v]}(\sum_{N\notin \mathscr{O}(W)}S_{M,N}^{V}Z_N(v,\tau)+\sum_{N\in \mathscr{O}(W)}S_{M,N}^{V}Z_N(v,\tau))\\
        =&\tau^{\mathrm{Wt}[v]}(\sum_{N\notin \mathscr{O}(W)}S_{M,N}^{V}Z_N(v,\tau)+\sum_{N\in \mathscr{O}(W)}S_{M,N}^{V}\sum_{j=1}^{n}(\mathrm{dim}\ U_j)Z_{W_j}(v,\tau))\\
        =&\tau^{\mathrm{Wt}[v]}(\sum_{N\notin \mathscr{O}(W)}S_{M,N}^{V}Z_N(v,\tau)+\sum_{j=1}^{n}((\mathrm{dim}\ U_j)\sum_{N\in \mathscr{O}(W)}S_{M,N}^{V}Z_{W_j}(v,\tau))).\\
    \end{align*}
    On the other hand,
    \begin{align*}
        Z_M(v,-\frac{1}{\tau})=&\sum_{i=1}^{m}Z_{M_i}(v,-\frac{1}{\tau})\\
        =&\sum_{i=1}^{m}\tau^{\mathrm{Wt}[v]}((\sum_{N\in \mathscr{M}_{V^G},\ N\neq W_j}S_{M_i,N}^{V^G}Z_N(v,\tau))+S_{M_i,W_j}^{V^G}Z_{W_j}(v,\tau))\\
        =&\tau^{\mathrm{Wt}[v]}((\sum_{i=1}^{m}\sum_{N\in \mathscr{M}_{V^G},\ N\neq W_j}S_{M_i,N}^{V^G}Z_N(v,\tau))+\sum_{i=1}^{m}S_{M_i,W_j}^{V^G}Z_{W_j}(v,\tau))\\
        =&\tau^{\mathrm{Wt}[v]}((\sum_{N\in \mathscr{M}_{V^G},\ N\neq W_j}\sum_{i=1}^{m}S_{M_i,N}^{V^G}Z_N(v,\tau))+S_{M,W_j}^{V^G}Z_{W_j}(v,\tau))\\
        =&\tau^{\mathrm{Wt}[v]}((\sum_{N\in \mathscr{M}_{V^G},\ N\neq W_j}S_{M,N}^{V^G}Z_N(v,\tau))+S_{M,W_j}^{V^G}Z_{W_j}(v,\tau)).
    \end{align*}
    A comparison of the coefficients of \(Z_{W_j}(v,\tau)\) shows that \(S_{M,W_j}^{V^G}=(\mathrm{dim}\ U_j)\sum_{N\in \mathscr{O}(W)}S_{M,N}^V\).
\end{proof}
\noindent The next corollary follows directly from the preceding lemma.
\begin{corollary}
    \label{corollaryCircg}
    Use the same notation in the preceding lemma. Then,
    \begin{equation*}
        S_{M,W_j}^{V^G}=(\mathrm{dim}U_j)\frac{\sum_{g\in G}S_{M,W\circ g}^V}{|G_W|}.
    \end{equation*}
\end{corollary}
\begin{proof}
    Equation \ref{equationSEntry} shows that
    \begin{align*}
        S_{M,W_j}^{V^G}=&(\mathrm{dim}U_j)\sum_{N\in \mathscr{O}(W)}S_{M,N}^V\\
        =&(\mathrm{dim}U_j)\frac{\sum_{g\in G}S_{M,W\circ g}^V}{|G_W|}.
    \end{align*}
\end{proof}
\noindent The formula in the next lemma is useful in our further discussions.
\begin{lemma}
    \label{lemmaCircCirc}
    Let \(M,W\in \mathscr{M}_V\), and \(g\in G\). Then,
    \begin{equation}
        S_{M\circ g, W\circ g}^{V}=S_{M, W}^{V}.
    \end{equation}
\end{lemma}
\begin{proof}
    Definition \ref{definitionExtension} shows that
    \begin{align*}
        Z_{M\circ g}(v,-\frac{1}{\tau})=&\tau^{\mathrm{Wt}[v]}\sum_{N\circ g\in \mathscr{M}_V}S_{M\circ g,N\circ g}^{V}Z_{N\circ g}(v,\tau).
    \end{align*}
    On the other hand,
    \begin{align*}
        Z_{M\circ g}(v,-\frac{1}{\tau})=&Z_M(gv,-\frac{1}{\tau})\\
        =&\tau^{\mathrm{Wt}[v]}\sum_{N\in \mathscr{M}_V}S_{M,N}^{V}Z_N(gv,\tau)\\
        =&\tau^{\mathrm{Wt}[v]}\sum_{N\in \mathscr{M}_V}S_{M,N}^{V}Z_{N\circ g}(v,\tau).\\
    \end{align*}
    A comparison of the coefficients of \(Z_{N\circ g}(v,\tau)\) shows that \(S_{M\circ g, N\circ g}^{V}=S_{M, N}^{V}\).
\end{proof}
\noindent The next corollary follows directly from the preceding lemma.
\begin{corollary}
    \label{corollaryCircCirc}
    Let \(W\in \mathscr{M}_V\), and \(g\in G\). Then,
    \begin{equation}
        S_{V, W\circ g}^{V}=S_{V, W}^{V}.
    \end{equation}
\end{corollary}
\begin{proof}
    Since \(g\) is an automorphism of \(V\), we have \(V\circ g\cong V\). Lemma \ref{lemmaCircCirc} shows that
    \begin{align*}
        S_{V, W\circ g}^{V}=&S_{V\circ g, W\circ g}^{V}\\
        =&S_{V, W}^{V}.
    \end{align*}
\end{proof}
\begin{corollary}
    \label{ccorollary2}
    Use the same notation in Corollary \ref{corollaryCircg}. Then,
    \begin{equation*}
        S_{V,W_j}^{V^G}=(\mathrm{dim}U_j)|G|\frac{S_{V,W}^V}{|G_W|}.
    \end{equation*}
\end{corollary}
\begin{proof}
    By Corollaries, \ref{corollaryCircg}, \ref{corollaryCircCirc}, we have
    \begin{align*}
        S_{V,W_j}^{V^G}=&(\mathrm{dim}U_j)\frac{\sum_{g\in G}S_{V,W\circ g}^V}{|G_W|} \\
        =&(\mathrm{dim}U_j)\frac{\sum_{g\in G}S_{V,W}^V}{|G_W|} \\
        =&(\mathrm{dim}U_j)|G|\frac{S_{M,W}^V}{|G_W|}.
    \end{align*}
\end{proof}
\noindent Formulas in the next lemma play an important role in our further discussions.
\begin{lemma}
    Let \(M,N,W\in \mathscr{M}_V\), such that \(W\) is not an irreducible \(V^{G}\)-module. By orbifold theory, denote \(W=\bigoplus_{i=1}^{n}W_{i}\otimes P_i\), where \(W_i\in \mathscr{M}_{V^G}\). Then,
    \begin{equation}
        \label{equationWu2}
        S_{\{(M\boxtimes_{V^G}N),W_i\}}^{V^G}=\frac{1}{S_{V^G,W_i}^{V^G}}\frac{(\mathrm{dim}\ P_i)^2}{|G_W|^2} (\sum_{g\in G}S_{M,W\circ g}^V)(\sum_{g\in G}S_{N,W\circ g}^V),
    \end{equation}
and
    \begin{equation}
        \label{equationWu3}
        S_{\{(\bigoplus_{g\in G} M\boxtimes_{V}(N\circ g)),W_i\}}^{V^G}=\frac{1}{S_{V,W}^{V^G}}\frac{\mathrm{dim}\ P_i}{|G_W|}(\sum_{g\in G}S_{M,W\circ g}^V)(\sum_{g\in G}S_{N,W\circ g}^V).
    \end{equation}
\end{lemma}
\begin{proof}
    Lemma \ref{lemmaBasicS-Fusion} and Corollary \ref{corollaryCircg} show that
    \begin{align*}
        S_{\{(M\boxtimes_{V^G}N),W_i\}}^{V^G}=&\frac{1}{S_{V^G,W_i}^{V^G}}{S_{M,W_i}^{V^G} S_{N,W_i}^{V^G}}\\
        =&\frac{1}{S_{V^G,W_i}^{V^G}} (\mathrm{dim}\ P_i)\frac{\sum_{g\in G}S_{M,W\circ g}^V}{|G_W|}(\mathrm{dim}\ P_i)\frac{\sum_{g\in G}S_{N,W\circ g}^V}{|G_W|}\\
        =&\frac{1}{S_{V^G,W_i}^{V^G}}\frac{(\mathrm{dim}\ P_i)^2}{|G_W|^2} (\sum_{g\in G}S_{M,W\circ g}^V)(\sum_{g\in G}S_{N,W\circ g}^V).
    \end{align*}
    Lemmas \ref{lemmaBasicS-Fusion}, \ref{lemmaCircCirc}, and Corollaries \ref{corollaryCircg}, \ref{corollaryCircCirc} show that
    \begin{align*}
        S_{\{(\bigoplus_{g\in G} M\boxtimes_{V}(N\circ g)),W_i\}}^{V^G}=& \sum_{g\in G} S_{\{(M\boxtimes_{V}(N\circ g)),W_i\}}^{V^G}\\
        =&\frac{\mathrm{qdim}_{V^G}W_i}{|G_W|}\sum_{g,h\in G} S_{\{(M\boxtimes_{V}(N\circ g)),W_i\circ h\}}^{V} \\
        =&\frac{\mathrm{dim}P_i}{|G_W|}\sum_{g,h\in G} \frac{1}{S_{V,W\circ h}^V}S_{M,W\circ h}^VS_{M\circ g, W\circ h}^V \\
        =&\frac{\mathrm{dim}P_i}{|G_W|}\sum_{g,h\in G} \frac{1}{S_{V,W}^V}S_{M,W\circ h}^VS_{M, W\circ g^{-1}h}^V\\
        =&\frac{1}{S_{V,W}}\frac{\mathrm{dim}P_i}{|G_W|}(\sum_{g\in G}S_{M,W\circ g}^V)(\sum_{g\in G}S_{N,W\circ g}^V).
    \end{align*}
\end{proof}
\begin{remark}
    \label{remarkTwisted}
    Let \(M,N\in \mathscr{M}_V\).
    Let \(W\in \mathscr{M}_{V^G, \mathrm{II}}\).
    Then, \(S_{M,W}^{V^G}=S_{N,W}^{V^G}=0\).
    This implies \(S_{\{(M\boxtimes_{V^G}N),W\}}^{V^G}=S_{\{(\bigoplus_{g\in G} M\boxtimes_{V}(N\circ g)),W\}}^{V^G}=0\).
\end{remark}
\noindent The formula in the next lemma will give an important property of the unitary space constructed in Section 4.
\begin{lemma}
    \label{lemmaSII=0}
    Let \(M\in \mathcal{S}_V(G)\), and \(F_i\in \mathscr{M}_{V^G,\mathrm{II}}\). Then,
    \begin{equation*}
        S_{M,F_i}^{V^G}=0.
    \end{equation*}
\end{lemma}
\begin{proof}
    Modular invariance shows that
    \begin{equation*}
        Z_{M}(v,-\frac{1}{\tau})=\tau^{\mathrm{Wt}[v]}\sum_{N\in \mathscr{M}_V}S_{M,N}^{V}Z_N(V,(1,g),\tau).
    \end{equation*}
    Reference \cite{2015arXiv150703306DOrbifold} indicates that
    \begin{equation*}
        Z_N(v,(1,g),\tau)=\sum_{\lambda\in \Lambda_{G_N,\alpha_N}} \lambda(g)Z_{N_{\lambda}}(v,\tau)
    \end{equation*}
    Since \(N\in \mathscr{M}_V\), \(N_{\lambda}\in \mathscr{M}_{V^G}\) is of type one. Thus,
    \begin{align*}
        Z_{M}(v,-\frac{1}{\tau})=&\tau^{\mathrm{Wt}[v]}\sum_{N\in \mathscr{M}_V}S_{M,N}^{V}Z_N(V,(1,g),\tau)\\
        =&\tau^{\mathrm{Wt}[v]}\sum_{N\in \mathscr{M}_V}(\sum_{\lambda\in \Lambda_{G_N,\alpha_N}}\lambda(g)Z_{N_{\lambda}}(\tau))\\
        =&\tau^{\mathrm{Wt}[v]}\sum_{M_i\in \mathscr{M}_{V^G, \mathrm{I}}}a_iZ_{M_i}(v,\tau),
    \end{align*}
    where \(a_i\in \mathbb{C}\). Since none of the modules in \(\mathscr{M}_{V^G,\mathrm{II}}\) appears in the right side of the preceding equation, we have \(S_{M,F_i}^{V^G}=0\), where \(F_i\in \mathscr{M}_{V^G,\mathrm{II}}\).
\end{proof}
\noindent The next corollary follows directly from the preceding lemma.
\begin{corollary}
    \label{corollarySMFi=0}
    Let \(M\in \mathscr{M}_{V^G}\), and \(N \in \mathcal{S}_{V}(G)\), and \(F_i\in \mathscr{M}_{V^G,\mathrm{II}}\). Then,
    \begin{equation*}
        S_{M\boxtimes_{V^G}N,F_i}^{V^G}=0.
    \end{equation*}
\end{corollary}
\begin{proof}
    Apply Lemmas \ref{lemmaBasicS-Fusion}, \ref{lemmaSMFi=0}. we have
    \begin{align*}
    S_{M\boxtimes_{V^G}N,F_i}^{V^G}=&\frac{1}{S_{V^G,F_i}^{V^G}}{S_{M,F_i}^{V^G}}S_{N,F_i}^{V^G}\\
    =&\frac{1}{S_{V^G,F_i}^{V^G}}{S_{M,F_i}^{V^G}}\times 0\\
    =&0.
    \end{align*}
\end{proof}
\noindent In the next lemma, we show a result for fusion products involving \(V\) over \(V^G\).
\begin{lemma}
    Let \(M\in \mathcal{S}_V(G)\). Then,
    \begin{equation*}
        V\boxtimes_{V^G}M=|G|M.
    \end{equation*}
\end{lemma}
\begin{proof}
    Assume \(F_i\in \mathscr{M}_{V^G,\mathrm{I}}\).
    Lemma \ref{lemmaBasicS-Fusion} and Equation \ref{equationWu4} show that
    \begin{align*}
        S_{V\boxtimes_{V^G}M,E_i}^{V^G}=&\frac{1}{S_{V^G,F_i}^{V^G}}{S_{V,F_i}^{V^G}}S_{M,F_i}^{V^G}\\
        =&\frac{1}{S_{V^G,F_i}^{V^G}}{|G|S_{V^G,F_i}^{V^G}}S_{M,F_i}^{V^G}\\
        =&|G|S_{M,F_i}^{V^G} \\
        =&S_{|G|M,F_i}^{V^G}.
    \end{align*}
    Let \(F_i\in \mathscr{M}_{V^G,\mathrm{II}}\). Lemma \ref{lemmaSMFi=0} and Corollary \ref{corollarySMFi=0} show that
    \begin{align*}
        S_{V\boxtimes_{V^G}M,F_i}^{V^G}=&0 \\
        =&|G|0\\
        =&|G|S_{M,F_i}^{V^G}\\
        =&S_{|G|M,F_i}^{V^G}.
    \end{align*}
    Thus, \(S_{V\boxtimes_{V^G}M,F_i}^{V^G}=S_{|G|M,F_i}^{V^G}\), for every \(F_i\in \mathscr{M}_{V^G}\). Hence, Lemma \ref{lemmaSameModule} implies \(V\boxtimes_{V^G}M=|G|M\) as \(V^G\)-modules.
\end{proof}

\newpage
\section{The unitary space structure on the set of \(V\)-modules}
\begin{remark}
    \label{remarkVectorSpace}
    In the first place, we introduce some notations used in this section.
    Let \(\mathbf{E}_V\) be the linear space spanned by \(\mathscr{M}_{V}\) over \(\mathbb{C}\).
    In the rest of the paper, write \(\mathbf{E}_{V^G}\) as \(\mathbf{E}\) for convenience.
    Then, \(\mathrm{dim}\mathbf{E}=|\mathscr{M}_{V^G}|\), and \(A(\sum a_iM_i, \sum b_iM_i)=S_{\sum a_iM_i, \sum b_iM_i}^{V^G}\) is a bilinear form on \(\mathbf{E}\).
    Since the \(S\)-matrix of \(V^G\) is unitary, \(A(\cdot,\cdot)\) is nondegenerate.
    Let \(\mathbf{W}\) be the linear subspace of \(\mathbf{E}\) spanned by \(\mathscr{M}_{V^G,\mathrm{II}}\).
    Let \(\mathbf{W}^A=\{x\in \mathbf{E}|A(x,y)=0,\ \forall y\in\mathbf{W}\}\) consisting of the annihilators of \(\mathbf{W}\).
    Then,
    \begin{align*}
        \mathrm{dim}(\mathbf{W}^A)=&\mathrm{dim}\mathbf{E}-\mathrm{dim}\mathbf{W}\\
        =&|\mathscr{M}_{V^G,\mathrm{I}}|.
    \end{align*}
    Let \(\mathcal{M}_V\) be the set of all \(V\)-modules. Define
    \begin{align*}
        \iota: \ \ \ \ \ \ \ \ \ \  \mathcal{M}_{V^G} \rightarrow& \mathbf{E}\\
        \iota(\bigoplus_i a_iM_i)=&\sum_i a_iM_i,
    \end{align*}
    where \(M_i\in \mathscr{M}_{V^G}\), and \(a_i\in \mathbb{C}\).
    Let \(\mathbf{U}\) be the linear subspace of \(\mathbf{E}\) spanned by \(\{\iota(M)|M\in \mathcal{S}_V(G)\}\). Then, \(\mathrm{dim}\mathbf{U}=|\mathscr{S}_V(G)|\).
    \begin{itemize}
    \item (a) If \(M\in \mathcal{S}_{V}(G)\), Lemma \ref{lemmaSMFi=0} shows that \(\iota(M)\in \mathbf{W}^A\).
    \item (b) If \(M\in \mathscr{M}_{V^G}\), and \(N \in \mathcal{S}_{V}(G)\), Corollary \ref{corollarySMFi=0} shows that \(\iota(M\boxtimes_{V^G}N)\in \mathbf{W}^A\).
    \item (c) \(\mathbf{U}=\mathbf{W}^A\).
    \item (d) \(\iota(M\boxtimes_{V^G}N)\in \mathbf{U}\), if and only if \(M\in \mathbf{U}\), or \(N\in \mathbf{U}\).
    \end{itemize}
    Statements (c), (d) will be proved in the following paragraphs.
\end{remark}
\begin{lemma}\label{lemmaUEqualsWA}
    Use the notations introduced in Remark \ref{remarkVectorSpace}. Then,
    \begin{equation*}
    \mathbf{U}=\mathbf{W}^A.
    \end{equation*}
\end{lemma}
\begin{proof}
    Let \(v\in V^G\), \(\tau\rightarrow -\frac{1}{\tau}\), \(\gamma=
    \begin{pmatrix}
        0&-1\\
        1&1
    \end{pmatrix}\),
    and \(N\in \mathscr{M}(1,g)\), where \(g\in G\).
    Then,
    \begin{equation}
        \label{equationDim1}
        Z_N(v,(1,g),-\frac{1}{\tau})=\tau^{\mathrm{Wt}[v]}\sum_{M_i\in \mathscr{M}_{V^G}}a_iZ_{M_i}(v,\tau).
    \end{equation}
    Modular invariance shows that
    \begin{equation}
        \label{equationDim2}
        Z_N(v,(1,g),-\frac{1}{\tau})=\tau^{\mathrm{Wt}[v]}\sum_{M\in \mathscr{M}(g)}S_{N,M}Z_{M}(v,\tau).
    \end{equation}
    Let \(M\in \mathscr{M}_V\). Use generalized Galois theory to denote \(M=\sum_{\lambda\in \Lambda_{G_M,\alpha_M}}W_{\lambda}\otimes M_{\lambda}\).
    By Reference \cite{2015arXiv150703306DOrbifold},
    \begin{equation}
        \label{equationDim3}
        Z_{M_\lambda}(v,-\frac{1}{\tau})=\frac{1}{|G_M|}\sum_{g\in G_M}Z_M(v,(1,g),-\frac{1}{\tau})\overline{\lambda(g)}.
    \end{equation}
    Since \(M_{\lambda}\in \mathscr{M}_{V^G,\mathrm{I}}\), Equations \ref{equationDim2}, \ref{equationDim3} show that
    \begin{equation}
        \label{inequalityDim1}
        |\mathscr{M}_{V^G,\mathrm{I}}|\leq |\mathcal{S}_{V}(G)|
    \end{equation}
    It follows that \(\mathrm{dim}\mathbf{W}^A\leq \mathrm{dim}\mathbf{U}\).
    Statement (a) in Remark \ref{remarkVectorSpace} implies \(\mathrm{dim}\mathbf{U}\leq \mathrm{dim}\mathbf{W}^A\), and \(\mathbf{U}\) is a subspace of \(\mathbf{W}^A\).
    Thus, \(\mathrm{dim}\mathbf{U}=\mathrm{dim}\mathbf{W}^A\), and hence \(\mathbf{U}=\mathbf{W}^A\).
\end{proof}
\begin{remark}
    Let \(M\in \mathcal{M}_{V^G}\) The preceding lemma implies \(S_{M,F_i}^{V^G}=0\), for every \(F_i \in \mathscr{M}_{V^G,\mathrm{II}}\),  if and only if \(M\in \mathbf{B}\).
\end{remark}
\begin{remark}\label{rm:length_M}
    Let \(\mathscr{M}_{V^G}=\{M^0,M^1,\ldots, M^d\}\), where \(d\in \mathbb{N}\).  Let \(e_i=\iota(M^i)\). Let \(\langle \cdot , \cdot \rangle\) be the Hermitian inner product on \(\mathbf{E}\), where \(\{e_0,e_1,\ldots,e_d\}\) is an orthonormal basis of \(\mathbf{E}\).
    By References \cite{verlinde1988fusion}, \cite{huang2008vertex}, the \(S\) matrix is unitary.
    This implies \(\{\epsilon_i\}\), where \(\epsilon_i=\sum_{k=0}^dS_{i,k}e_k\) is also an orthonormal basis of \(\mathbf{E}\). Let \(M\in \mathscr{M}_{V}(g)\). The one-to-one map \(\iota\) is omitted, if no confusion follows. By orbifold theory, denote \(M=\bigoplus_{\lambda\in \Lambda_{G_M,\alpha_M}}W_{\lambda}\otimes M_{\lambda}\), where \(M_{\lambda}\in \mathscr{M}_{V^G}\). Then,
    \begin{align*}
        \langle M, M \rangle=&\langle \bigoplus_{\lambda\in \Lambda_{G_M,\alpha_M}}W_{\lambda} \otimes M_{\lambda}, \bigoplus_{\lambda\in \Lambda_{G_M,\alpha_M}}W_{\lambda} \otimes M_{\lambda} \rangle \\
                            =&\langle \sum_{\lambda\in \Lambda_{G_M,\alpha_M}}(\mathrm{dim}W_{\lambda})M_{\lambda}, \sum_{\lambda\in \Lambda_{G_M,\alpha_M}}(\mathrm{dim}W_{\lambda})M_{\lambda} \rangle \\
                            =&\sum_{\lambda\in \Lambda_{G_M,\alpha_M}}(\mathrm{dim}W_{\lambda})^2 \\
                            =&|G_M|.
    \end{align*}
By \cite[Theorem 3.2]{2015arXiv150703306DOrbifold}, with \(N\in \mathscr{M}_{V}(h)\) , we have
\begin{equation}\label{eq:orth_TModule}
 \langle M, N \rangle=\delta_{\mathscr{O}(M),\mathscr{O}(N)}|G_M|.
\end{equation}
\end{remark}
\noindent As an application, we show \(G_M=G\) under some assumptions.
\begin{proposition}\label{propo:GmEqualG}
    Let \(G\) be a cyclic group, and \(V\) a vertex operator algebra.
    Assume \(G_{M_k}=G\), or \(\{e\}\), for \(k=1, 2, \ldots, n\).
    Then, \(G_M=G\), for each \(M\in \mathscr{M}_V(g)\), where \(g\neq e\).
\end{proposition}
\begin{proof}
    Let \(\mathscr{M}_V=\{M_1, M_2, M_3, \ldots, M_n\}\).
    Let \(\mathscr{S}_V(\{e\})=\mathscr{M}_{V}/G=\{S_1, S_2, \ldots, S_l\}\).
    Then we have
    \begin{align*}
        \sum_{g\neq e}|\mathscr{M}_V(g)|&\geq |(\bigcup_{g\neq e}\mathscr{M}_V(g))/G|\\
        &=|\mathscr{S}_V(G)|-l .\\
    \end{align*}
By lemma \ref{lemmaUEqualsWA}, we have
\begin{equation}
\sum_{g\neq e}|\mathscr{M}_V(g)|\geq |\dim(W^A)|-l=|\mathscr{M}_{V^G,I}|-l.
\end{equation}
By theorem \ref{theoremModularInvariance}, we have
$\mathscr{M}_V(g)=|\{M|g\in G_M\,M\in \mathscr{M}_V\}|$.
It follows that
\begin{align*}
\sum_{g\in G}|\mathscr{M}_V(g)|&=\sum_{g\in G}|\{M|g\in G_M,M\in \mathscr{M}_V\}|\\
&=|\{(g,M)|g\in G,g\in G_M,M\in \mathscr{M}_V\}|\\
&=\sum_{M\in \mathscr{M}_V}|G_M| .
\end{align*}
Note that $\mathscr{O}(M)=|G/G_M|$, we have
\[
\sum_{g\neq e}|\mathscr{M}_V(g)|=\sum_{g\in G}|\mathscr{M}_V(g)|-n=|G||\mathscr{M}_V/G|-n=l|G|-n.
\]
According to a well known fact that any projective representation of a cyclic group is ordinary, we see that each $M\in\mathscr{M}_V$ is
a $G_M$ module.
Let  $p=|\{M|G_M=G,M\in \mathscr{M}_V\}|$,$q=|\{M|G_M={e},M\in \mathscr{M}_V\}|$.
By theorem \ref{theoremGeneralizedQuantumGaloisTheory}, we have
$M=\oplus_{\lambda\in {\rm irr}(G_M)}W_\lambda\otimes M_{\lambda}$, and
$|\mathscr{M}_{V^G,I}|=p|{\rm irr}(G)|+q/|G|=p|G|+q/|G|$.
Then we have
\begin{equation}
|\mathscr{M}_{V^G,I}|-l=p|G|+q/|G|-p-q/|G|=l|G|-p-q=l|G|-n.
\end{equation}
It follows that
\[
\sum_{g\neq e}|\mathscr{M}_V(g)|= |(\bigcup_{g\neq e}\mathscr{M}_V(g))/G|.
\]
Then the proposition follows.
\end{proof}
As an application, we give a corollary.
\begin{corollary}
If $V$ has only two irreducible modules and $G$ is a cyclic group, then for each $M\in \mathscr{M}_V(G)$ we have $G_M=G$.
\end{corollary}
\begin{proof}
Let $\mathscr{M}_V=\{M_0=V,M_1\}$.
Note that $M_1\circ g\neq V$. It follows that $M_1\circ g= M_1$.
Then we have $G_{M_0}=G_{M_1}=G$. 
By proposition \ref{propo:GmEqualG}, we have $G_{M}=G$ for each $M\in\mathscr{M}_V(G)$.
\end{proof}
\newpage
\section{Fusion products of irreducible type-one \(V^G\)-modules}
\noindent In this section, let \(N\) be in \(\mathscr{M}_{V^G,\mathrm{I}}\).
Denote the quantum Galois decomposition of \(V\) by \(V=\bigoplus_{\chi\in \mathrm{Irr}(G)}W_{\chi}\otimes V_{\chi}\), where
\(V_{\chi}\) is an irreducible \(V^G\)-module, \(W_{\chi}\) is an irreducible \(\mathbb{C}G\)-module, and \(\mathrm{qdim}V_{\chi}=\mathrm{dim}W_{\chi}\). Write \(\mathscr{M}_{V^G}\) as \(\{M^0,M^1,M^2,\cdots, M^d\}\), where \(M^0\cong V^G\).
\begin{theorem}
    \label{theoremMainConstant}
    Let \(M\in \mathcal{S}_{V}(G)\), and \(E\in\mathscr{M}_{V^G,\mathrm{I}}\).
    By orbifold theory, denote \(M=\bigoplus_{\lambda\in \Lambda_{G_M,\alpha_M}}W_{\lambda}\otimes M_{\lambda}\), where \(M_{\lambda}\in \mathscr{M}_{V^G}\).
    Then,
    \begin{equation}
    \label{equationMainQuotientConstant}
        \frac{S_{M_{\lambda_1},E}}{\mathrm{qdim}_{V^G}M_{\lambda_1}}= \frac{S_{M_{\lambda_2},E}}{\mathrm{qdim}_{V^G}M_{\lambda_2}},
    \end{equation}
    where \(\lambda_1,\lambda_2 \in \Lambda_{G_M,\alpha_M}\).
\end{theorem}
\begin{proof}
    Standard modular invariance formula shows that
    \begin{equation}
        \label{equationMainTheorem1}
        Z_{M_{\lambda}}(v,-\frac{1}{\tau})=\tau^{\mathrm{wt}[v]}\sum_{W\in \mathscr{M}_{V^G}}S_{M_{\lambda},W}^{V^G}Z_W(v,\tau).
    \end{equation}
    Reference \cite{2015arXiv150703306DOrbifold} implies the following relation
    \begin{align}
        \label{equationMainTheorem2}
        Z_{M_\lambda}(v,-\frac{1}{\tau})=&\frac{1}{|G_M|}\tau^{\mathrm{wt}[v]}\sum_{h\in G_M}\sum_{N\in \mathscr{M}(h,g^{-1})}S_{M,N}^VZ_N(v,(h,g^{-1}),\tau)\overline{\lambda(h)}.
    \end{align}
    Orbifold theory implies that
    \begin{equation}
        \label{equationMainTheorem3}
        Z_{N}(v,(h,g^{-1}),\tau)=\sum_{\mu\in \Lambda_{G_N,\alpha_N}}\mu (h)Z_{N_{\mu}}(v,\tau).
    \end{equation}
    Plug Equation \ref{equationMainTheorem3} in Equation \ref{equationMainTheorem2}.
    \begin{align}
        \label{equationMainTheorem4}
        Z_{M_\lambda}(v,-\frac{1}{\tau})=&\frac{1}{|G_M|}\tau^{\mathrm{wt}[v]}\sum_{h\in G_M}\sum_{N\in \mathscr{M}(h,g^{-1})}S_{M,N}^V\sum_{\mu\in \Lambda_{G_N,\alpha_N}}\mu (h)Z_{N_{\mu}}(v,\tau)\overline{\lambda(h)}.
    \end{align}
    Then, for \(h=1\), observe the coefficients of \(Z_{N_{\mu}}(v,\tau)\) in Equations \ref{equationMainTheorem1}, \ref{equationMainTheorem4}. A comparison of these two coefficients yields
    \begin{equation*}
        S_{M_{\lambda},N_{\mu}}^{V^G}=\frac{1}{|G_M|}S_{M,N}^V\mu(1)\overline{\lambda(1)}.
    \end{equation*}
    Thus, by Equation \ref{equationOrbifoldProportional}, we have
    \begin{align*}
        \frac{S_{M_{\lambda_1},N_{\mu}}^{V^G}}{S_{M_{\lambda_2},N_{\mu}}^{V^G}}=& \frac{\overline{\lambda_1(1)}}{\overline{\lambda_2(1)}} \\
        =&\frac{\mathrm{dim}W_{\lambda_1}}{\mathrm{dim}W_{\lambda_2}} \\
        =&\frac{\mathrm{qdim}_{V^G}M_{\lambda_1}}{\mathrm{qdim}_{V^G}M_{\lambda_2}},
    \end{align*}
    where \(\lambda_1, \lambda_2 \in \Lambda_{G_M,\alpha_M}\). The desired result follows.
\end{proof}
\begin{corollary}
    Let \(U\), \(N\) be irreducible \(V\)-modules. Assume that \(U\) is not irreducible as a \(V^{\langle g \rangle}\)-module. Denote \(U=\bigoplus_{i=0}^{T-1}U^{i}\). Then, \(S_{U^{i},N}=S_{U^{j},N}\).
\end{corollary}
\begin{proof}
    Notice that \(\mathrm{qdim}_{V^G}U^i=\mathrm{qdim}_{V^G}U^j\), because \(\langle g \rangle\) is a cyclic group. So, Theorem \ref{theoremMainConstant} implies the desired result.
\end{proof}
\begin{theorem}
Let \(M\in \mathcal{S}_V(G)\), and \(W\in \mathscr{M}_V\).
By orbifold theory, denote \(M=\bigoplus_{\lambda\in \Lambda_{G_M,\alpha_M}}Q_{\lambda}\otimes M_{\lambda}\), and \(W=\bigoplus_{\mu\in \Lambda_{G_W,\alpha_W}}P_{\mu}\otimes W_{\mu}\). Then,
\begin{equation}
    \label{equationComputationMi}
    S_{M_{\lambda},W_\mu}^{V^G}=\frac{\mathrm{dim}\ Q_{\lambda}}{|G_M|}(\mathrm{dim}P_{\mu})\sum_{N\in \mathscr{O}(W)}S_{M,N}^V.
\end{equation}
\end{theorem}
\begin{proof}
Since \(M=\bigoplus_{\lambda\in \Lambda_{G_M,\alpha_M}}Q_{\lambda}\otimes M_{\lambda}\), we have
\begin{align*}
    S_{M,W_{\mu}}^{V^G}=\sum_{\lambda \in \Lambda_{G_M,\alpha_M}}(\mathrm{dim} Q_{\lambda})S_{M_{\lambda},W_\mu}^{V^G}.
\end{align*}
On the other hand, by \(W=\bigoplus_{\mu\in \Lambda_{G_W,\alpha_W}}P_{\mu}\otimes W_{\mu}\) and Lemma \ref{lemmaOrbits}, we have
\begin{equation*}
    S_{M,W_{\mu}}^{V^G}=(\mathrm{dim}P_{\mu})\sum_{N\in \mathscr{O}(W)}S_{M,N}^V.
\end{equation*}
It follows that
\begin{equation*}
    \sum_{\lambda \in \Lambda_{G_M,\alpha_M}}(\mathrm{dim} Q_{\lambda})S_{M_{\lambda},W_\mu}^{V^G}=(\mathrm{dim}P_{\mu})\sum_{N\in \mathscr{O}(W)}S_{M,N}^V.
\end{equation*}
Let \(\lambda\) run over \(\Lambda_{G_M,\alpha_M}\). The preceding equation becomes a system of linear equations.
By Theorem \ref{theoremMainConstant}, one could solve this system of linear equations, and obtain
\begin{align*}
    S_{M_{\lambda},W_{\mu}}^{V^G}=&\frac{\mathrm{dim} Q_{\lambda}}{\sum_{\lambda \in \Lambda_{G_M,\alpha_M}}(\mathrm{dim} Q_{\lambda})^2}(\mathrm{dim} P_{\mu})\sum_{N\in \mathscr{O}(W)}S_{M,N}^V\\
    =&\frac{\mathrm{dim} Q_{\lambda}}{|G_M|}(\mathrm{dim} P_{\mu})\sum_{N\in \mathscr{O}(W)}S_{M,N}^V.
\end{align*}
\end{proof}
\begin{remark}
    In Equation \ref{equationComputationMi}, substitute \(V\) for \(M\). Recall the fact \(G_V=G\). By quantum Galois theory, denote \(V=\bigoplus_{\chi\in \mathrm{Irr}(G)}(Q_{\chi}\otimes V_{\chi})\). Equation \ref{equationComputationMi} becomes
    \begin{equation}
        \label{equationComputationVi}
        S_{V_{\chi},W_{\mu}}^{V^G}=\frac{(\mathrm{dim} Q_{\chi})}{|G|}(\mathrm{dim}P_{\mu})\sum_{N\in \mathscr{O}(W)}S_{V,N}^V.
    \end{equation}
\end{remark}
\begin{corollary} Let \(W\in \mathscr{M}_{V^G}\).
    \label{corollary1}
    \begin{equation}
        \label{equationWu4}
        S_{V,W}^{V^G}=|G|S_{V^G,W}^{V^G}.
    \end{equation}
\end{corollary}
\begin{proof}
    Notice that \(\mathrm{qdim}(V)=|G|\mathrm{qdim}(V^G)\). Theorem \ref{theoremMainConstant} implies the  desired result.
\end{proof}
\begin{theorem}
    \label{theoremMainFusion}
    Let \(M,N\in \mathscr{M}_V\), and \(g\in G\). Then,
    \begin{equation}
        M\boxtimes_{V^G} N=\bigoplus_{g\in G} M\boxtimes_{V}(N\circ g).
    \end{equation}
\end{theorem}
\begin{proof}
    By Corollary \ref{corollary1}, \ref{ccorollary2}, we have
    \begin{align*}
        |G|S_{V^G,W_i}^{V^G}=&S_{V,W_i}^{V^G}\\
        =&(\mathrm{dim}P_i)|G|\frac{S_{M,W}^V}{|G_W|}.
    \end{align*}
    That is,
    \begin{align}
        \label{equationJunFen}
        S_{V^G,W_i}^{V^G}=\frac{(\mathrm{dim}P_i)}{|G_W|}S_{V,W}^V.
    \end{align}
    This implies the right side of Equation \ref{equationWu2} equals the right side of Equation \ref{equationWu3}.
    It follows that \(S_{\{(M\boxtimes_{V^G}N),I\}}^{V^G}=S_{\{(\bigoplus_{g\in G} M\boxtimes_{V}(N\circ g)),I\}}^{V^G}=0\) for \(I\in \mathscr{M}_{V^G}\) of type one (untwisted).
    Remark \ref{remarkTwisted} shows that \(S_{\{(M\boxtimes_{V^G}N),I\}}^{V^G}=S_{\{(\bigoplus_{g\in G} M\boxtimes_{V}(N\circ g)),I\}}^{V^G}=0\) for \(I\in \mathscr{M}_{V^G}\) of type two (twisted).
    Therefore, \(S_{\{(M\boxtimes_{V^G}N),I\}}^{V^G}=S_{\{(\bigoplus_{g\in G} M\boxtimes_{V}(N\circ g)),I\}}^{V^G}=0\) for every \(I\in \mathscr{M}_{V^G}\). By Lemma \ref{lemmaSameModule}, \(M\boxtimes_{V^G} N=\bigoplus_{g\in G} M\boxtimes_{V}(N\circ g)\).
\end{proof}
\begin{lemma}
    \begin{equation}
        \label{equationWu1}
        V\boxtimes_{V^G}V=|G|V.
    \end{equation}
\end{lemma}
\begin{proof}
     This is a special case of Theorem \ref{theoremMainFusion}.
\end{proof}
\section{Fusion rules for \(\mathscr{M}_{V^G}\), and their quantum dimensions}

\begin{theorem}
    \label{theoremMainProportional}
    Let \(M,N,F\in \mathcal{S}_V(G)\). By orbifold theory, write
    \begin{align*}
        M=&\bigoplus_{\lambda\in\Lambda_{G_M,\alpha_M}}U_{\lambda}\otimes M_{\lambda} \\
        N=&\bigoplus_{\chi\in\Lambda_{G_N,\alpha_N}}W_{\chi}\otimes N_{\chi} \\
        F=&\bigoplus_{\xi\in\Lambda_{G_F,\alpha_F}}V_{\xi}\otimes F_{\xi}.
    \end{align*}
    Then,
    \begin{align*}
        \frac{\langle (M_{\lambda_1}\boxtimes_{V^G}N_{\chi_1}), F \rangle}{\langle (M_{\lambda_2}\boxtimes_{V^G}N_{\chi_2}), F \rangle}=\frac{(\mathrm{qdim}_{V^G}M_{\lambda_1})(\mathrm{qdim}_{V^G}N_{\chi_1})} {(\mathrm{qdim}_{V^G}M_{\lambda_2})(\mathrm{qdim}_{V^G}N_{\chi_2})},
    \end{align*}
    where \(\lambda_1,\lambda_2\in \Lambda_{G_M,\alpha_M}\), and \(\chi_1,\chi_2\in \Lambda_{G_N,\alpha_N}\).
\end{theorem}
\begin{proof}
    Note that the \(S\)-matrix is unitary. By Lemma \ref{lemmaSII=0} \(S_{F,E_i}^{V^G}=0\), for \(E_i\in \mathscr{M}_{V^G,II}\). Thus,
    \begin{align*}
        \langle M, F\rangle_{V^G}=&\sum_{E_i\in \mathscr{M}_{V^G}}S_{M,E_i}^{V^G}\overline{S_{F,E_i}^{V^G}} \\
        =&\sum_{E_i\in \mathscr{M}_{V^G,\mathrm{I}}}S_{M,E_i}^{V^G}\overline{S_{F,E_i}^{V^G}}.
    \end{align*}
    Apply Corollary \ref{corollarySMFi=0}, Lemma \ref{lemmaBasicS-Fusion}, and Equation \ref{equationJunFen}
    \begin{align*}
        \langle (M_{\lambda_1}\boxtimes_{V^G}N_{\chi_1}), F \rangle
        =&\sum_{E_i\in \mathscr{M}_{V^G,\mathrm{I}}}S_{(M_{\lambda_1}\boxtimes_{V^G}N_{\chi_1}),E_i}^{V^G} \overline{S_{F,E_i}^{V^G}}\\
        =&\sum_{E_i\in \mathscr{M}_{V^G,\mathrm{I}}}\frac{1}{S_{V^G,E_i}^{V^G}}S_{M_{\lambda_1},E_i}^{V^G} S_{N_{\chi_1},E_i}^{V^G} \overline{S_{F,E_i}^{V^G}}\\
        =&\sum_{E_i\in \mathscr{M}_{V^G,\mathrm{I}}}\frac{1}{S_{V^G,E_i}^{V^G}}\frac{\mathrm{dim}U_{\lambda_1}}{|G_M|} S_{M,E_i}^{V^G} \frac{\mathrm{dim}W_{\chi_1}}{|G_N|}S_{N,E_i}^{V^G} \overline{S_{F,E_i}^{V^G}}.
    \end{align*}
    Notice that the value, \(\frac{1}{S_{V^G,E_i}^{V^G}}\frac{1}{|G_M|} S_{M,E_i}^{V^G} \frac{1}{|G_N|}S_{N,E_i}^{V^G} \overline{S_{F,E_i}^{V^G}}\), is independent of choices of \(\lambda_1\) and \(\chi_1\). By Equation \ref{equationOrbifoldProportional}, the desired result follows
        \begin{align*}
        \frac{\langle (M_{\lambda_1}\boxtimes_{V^G}N_{\chi_1}), F \rangle}{\langle (M_{\lambda_2}\boxtimes_{V^G}N_{\chi_2}), F \rangle}=&\frac{(\mathrm{dim}U_{\lambda_1})(\mathrm{dim}W_{\chi_1})} {(\mathrm{dim}U_{\lambda_2})(\mathrm{dim}W_{\chi_2})} \\
        =&\frac{(\mathrm{qdim}_{V^G}M_{\lambda_1})(\mathrm{qdim}_{V^G}N_{\chi_1})} {(\mathrm{qdim}_{V^G}M_{\lambda_2})(\mathrm{qdim}_{V^G}N_{\chi_2})}.
    \end{align*}
\end{proof}
\begin{remark}
    Theorem \ref{theoremMainProportional} shows that
    \begin{align*}
        \langle (M_{\lambda_1}\boxtimes_{V^G}N_{\chi_1}), F \rangle=& \langle \sum_{\xi\in \Lambda_{G_F,\alpha_F}}N_{M_{\lambda_1},N_{\chi_1}}^{F_{\xi}}F_{\xi}, \sum_{\xi\in \Lambda_{G_F,\alpha_F}}(\mathrm{dim}V_{\xi})F_{\xi}\rangle \\
        =&\sum_{\xi\in \Lambda_{G_F,\alpha_F}}N_{M_{\lambda_1},N_{\chi_1}}^{F_{\xi}} (\mathrm{dim}V_{\xi}).
    \end{align*}
    By Equation \ref{equationOrbifoldProportional} and Theorem \ref{theoremMainProportional}, it follows that
    \begin{align*}
        \frac{\langle (M_{\lambda_1}\boxtimes_{V^G}N_{\chi_1}), F \rangle}{\langle (M_{\lambda_2}\boxtimes_{V^G}N_{\chi_2}), F \rangle}=& \frac{\sum_{\xi\in \Lambda_{G_F,\alpha_F}}N_{M_{\lambda_1},N_{\chi_1}}^{F_{\xi}} (\mathrm{dim}V_{\xi})} {\sum_{\xi\in \Lambda_{G_F,\alpha_F}}N_{M_{\lambda_2},N_{\chi_2}}^{F_{\xi}} (\mathrm{dim}V_{\xi})} \\
        =&\frac{\sum_{\xi\in \Lambda_{G_F,\alpha_F}}N_{M_{\lambda_1},N_{\chi_1}}^{F_{\xi}} (\mathrm{qdim}F_{\xi})} {\sum_{\xi\in \Lambda_{G_F,\alpha_F}}N_{M_{\lambda_2},N_{\chi_2}}^{F_{\xi}} (\mathrm{qdim}F_{\xi})} \\
        =&\frac{(\mathrm{qdim}_{V^G}M_{\lambda_1})(\mathrm{qdim}_{V^G}N_{\chi_1})} {(\mathrm{qdim}_{V^G}M_{\lambda_2})(\mathrm{qdim}_{V^G}N_{\chi_2})},
    \end{align*}
    This means, inside every irreducible \(g\)-twisted \(V\)-module \(F\), fusion rules of irreducible \(V^G\)-modules distribute "proportionally" to the product of their quantum dimensions.
\end{remark}

\begin{corollary}\label{coro:qdim_partial}
Use the same notations in Theorem \ref{theoremMainProportional}. We have
\begin{equation*}
  \qdim_{V^G}(\bigoplus_{\xi\in \Lambda_{G_F,\alpha_F}}N_{M_{\lambda_1},N_{\chi_1}}^{F_{\xi}}(F_\xi))
  =\frac{\dim(W_{\chi_1})\dim(U_{\lambda_1})}{|G_M||G_N|}\frac{\qdim_V(F)}{|G_F|}\sum_{g,h,l\in G}N^{F\circ h}_{M\circ l,N\circ g}
\end{equation*}
\end{corollary}
\begin{proof}
Note that
\begin{equation*}
 \sum_{\chi\in \Lambda_{G_N,\alpha_N},\lambda\in \Lambda_{G_M,\alpha_M}} \dim(W_{\chi})^2\dim(U_{\lambda})^2=|G_M||G_N|.
\end{equation*}
Recall that
\begin{equation*}
 M\boxtimes_{V^G}N=\sum_{\chi\in \Lambda_{G_N,\alpha_N},\lambda\in \Lambda_{G_M,\alpha_M}}\dim(W_{\chi})\dim(U_{\lambda}) M_\lambda\boxtimes_{V^G} N_\chi.
\end{equation*}
By Theorem \ref{theoremMainProportional}, it follows that
 \begin{equation*}
  \sum_{\xi\in \Lambda_{G_F,\alpha_F}}N_{M_{\lambda_1},N_{\chi_1}}^{F_{\xi}}(\dim(V_\xi))
  =\frac{\dim(W_{\chi_1})\dim(U_{\lambda_1})\langle M\boxtimes_{V^G}N, F \rangle}{|G_M||G_N|}\sum_{g,h\in G}N^{F\circ h}_{M,N\circ g}.
\end{equation*}
Note that
\begin{equation*}
 \qdim_{V^G}(F_\xi)=\dim V_\xi \qdim_V(F)\frac{|G|}{|G_F|},
\end{equation*}
and hence
\begin{align*}
  \qdim_{V^G}(\bigoplus_{\xi\in \Lambda_{G_F,\alpha_F}}N_{M_{\lambda_1},N_{\chi_1}}^{F_{\xi}}(F_\xi))
  &=\sum_{\xi\in \Lambda_{G_F,\alpha_F}}N_{M_{\lambda_1},N_{\chi_1}}^{F_{\xi}}(\dim(V_\xi))\qdim_V(F)\frac{|G|}{|G_F|}\\
  &=\frac{\dim(W_{\chi_1})\dim(U_{\lambda_1})\langle M\boxtimes_{V^G}N, F \rangle}{|G_M||G_N|}\frac{\qdim_V(F)|G|}{|G_F|}.
\end{align*}
So, by Theorem \ref{theoremMainFusion} and Equation \eqref{eq:orth_TModule}, we have
\begin{align*}
 \langle M\boxtimes_{V^G}N, F \rangle&=\sum_{g\in G}\langle M\boxtimes_{V}N\circ g, F \rangle\\
 &=\sum_{g\in G}\langle \sum_{E\in \mathscr{M}_{V}}N^{E}_{M,N\circ g}E,F\rangle \\
 &=\sum_{g\in G}\langle \sum_{E\in \mathscr{O}(F)}N^{E}_{M,N\circ g}E,F\rangle \\
 &=\sum_{g,h\in G}\langle N^{F\circ h}_{M,N\circ g}F\circ h,F\rangle \frac{|\mathscr{O}(F)|}{|G|}.
\end{align*}
Note that $\iota (F\circ h)= \iota (F)$ as vectors in $\mathbf{E}_{V^G}$.
By Remark \ref{rm:length_M}, we have
\begin{align*}
 \langle M\boxtimes_{V^G}N, F \rangle=&\sum_{g,h\in G}N^{F\circ h}_{M,N\circ g}\langle F,F\rangle \frac{|\mathscr{O}(F)|}{|G|}\\
 =&\sum_{g,h\in G}N^{F\circ h}_{M,N\circ g}|G_F|\frac{|\mathscr{O}(F)|}{|G|}\\
 =&\sum_{g,h\in G}N^{F\circ h}_{M,N\circ g}.
\end{align*}
Lemma \ref{lemmaCircCirc}, and Verlinde formula implies $N_{M,N}^F=N^{F\circ g}_{M\circ g,N\circ g}$. So, it follows that
\begin{equation*}
 |G|\sum_{g,h\in G}N^{F\circ h}_{M,N\circ g}=\sum_{g,h,l\in G}N^{F\circ h}_{M\circ l,N\circ g}.
\end{equation*}
Then, we have the desired result.
\end{proof}
\begin{remark}
\label{remarkInequality}
Use the same notations in Reference \cite{tanabe2005intertwining}, assume that $M,N,F\in \mathscr{M}_V$.
Set
\[
 \mathcal{I}=\bigoplus_{(L^1,L^2,L^3)\in \mathscr{O}(M)\times\mathscr{O}(N)\times\mathscr{O}(F)}I_V\begin{pmatrix}
                                                                                                    L^3\\
                                                                                                    \begin{matrix}
                                                                                                    L^1&L^2
                                                                                                    \end{matrix}
                                                                                                    \end{pmatrix}\otimes L^1\otimes L^2.
\]
Let $X^1=\Ind^{D(M)}_{S(M)}U_{\lambda_1}$, and $X^2=\Ind^{D(N)}_{S(N)}W_{\chi_1}$.
Set
\[
 \mathcal{I}_{\lambda_1,\chi_1}(v^1,v^2)
 =\Span_{\mathbb{C}}\{f\otimes(w^1\otimes v^1)\otimes(w^2\otimes v^2)\in \mathcal{I}|w^1\in X^1,w^2\in X^2\}
\]
with $v^1\in M_{\lambda_1}$,$v^2\in N_{\chi_1}$.
Then by Reference \cite[Theorem 2]{tanabe2005intertwining}, we have
\begin{equation}\label{eq:inequality_fusion}
 N^{F_{\xi}}_{M_{\lambda_1},N_{\chi_1}}\geq
 \dim_{\mathbb{C}}\Hom_{\mathcal{A}_{\alpha_{\mathscr{O}(F)}}(G,\mathscr{O}(F))}(\Ind^{D(F)}_{S(F)}V_{\xi},\mathcal{I}_{\lambda_1,\chi_1}(v^1,v^2)).
\end{equation}
Under certain assumption, this inequality is an equation.
\end{remark}
\begin{theorem}
Let $M,N,F\in\mathscr{M}_V$. Use the same notations in Remark \ref{remarkInequality}.
We have
\[
 N^{F_{\xi}}_{M_{\lambda_1},N_{\chi_1}}=
 \dim_{\mathbb{C}}\Hom_{\mathcal{A}_{\alpha_{\mathscr{O}(F)}}(G,\mathscr{O}(F))}(\mathrm{Ind}^{D(F)}_{S(F)}V_{\xi},\mathcal{I}_{\lambda_1,\chi_1}(v^1,v^2)).
\]
 \end{theorem}
\begin{proof}
Let $H^{F_\xi}_{M_{\lambda_1},N_{\chi_1}}
=\dim_{\mathbb{C}}\Hom_{\mathcal{A}_{\alpha_{\mathscr{O}(F)}}(G,\mathscr{O}(F))}(\Ind^{D(F)}_{S(F)}V_{\xi},\mathcal{I}_{\lambda_1,\chi_1}(v^1,v^2))$.
Note that \[\qdim_{V^G}(F_\xi)=\dim V_\xi \qdim_V(F)\frac{|G|}{|G_F|}.\]
We have
 \begin{align*}
  \qdim_{V^G}(\bigoplus_{\xi\in \Lambda_{G_F,\alpha_F}}H_{M_{\lambda_1},N_{\chi_1}}^{F_{\xi}}(F_\xi))=
  \sum_{\xi\in \Lambda_{G_F,\alpha_F}}H_{M_{\lambda_1},N_{\chi_1}}^{F_{\xi}}\dim(\Ind^{D(F)}_{S(F)}V_\xi)\qdim_V(F).
 \end{align*}
Recall that \( \mathcal{A}_{\alpha_{\mathscr{O}(F)}}(G,\mathscr{O}(F)) \) is semisimple, and its simple modules are precisely $\Ind^{D(F)}_{S(F)}V_\xi$.
It follows that
\begin{equation*}
 \sum_{\xi\in \Lambda_{G_F,\alpha_F}}H_{M_{\lambda_1},N_{\chi_1}}^{F_{\xi}}\dim(\Ind^{D(F)}_{S(F)}V_\xi)\qdim_V(F)
 =\qdim_V(F)\dim(\mathcal{I}_{\lambda_1,\chi_1}(v^1,v^2)).
\end{equation*}
Note that
\begin{equation*}
\mathcal{I}_{\lambda_1,\chi_1}(v^1,v^2)\cong \bigoplus_{(L^1,L^2,L^3)\in \mathscr{O}(M)\times\mathscr{O}(N)\times\mathscr{O}(F)}I_V\begin{pmatrix}
                                                                                                    L^3\\
                                                                                                    \begin{matrix}
                                                                                                    L^1&L^2
                                                                                                    \end{matrix}
                                                                                                    \end{pmatrix}\otimes U_{\lambda_1}\otimes W_{\chi_1}.
\end{equation*}
So, we have
\begin{align*}
 \dim(\mathcal{I}_{\lambda_1,\chi_1}(v^1,v^2))&=\sum_{(L^1,L^2,L^3)\in \mathscr{O}(M)\times\mathscr{O}(N)\times\mathscr{O}(F)}\dim I_V\begin{pmatrix}
                                                                                                    L^3\\
                                                                                                    \begin{matrix}
                                                                                                    L^1&L^2
                                                                                                    \end{matrix}
                                                                                                    \end{pmatrix}\dim U_{\lambda_1}\dim W_{\chi_1}\\
 &=\sum_{(L^1,L^2,L^3)\in \mathscr{O}(M)\times\mathscr{O}(N)\times\mathscr{O}(F)}F^{L^3}_{L_1,L_2}\dim U_{\lambda_1}\dim W_{\chi_1}\\
 &=\frac{\dim(W_{\chi_1})\dim(U_{\lambda_1})}{|G_M||G_N||G_F|}\sum_{g,h,l\in G}N^{F\circ h}_{M\circ l,N\circ g}.
\end{align*}
It follows that
\begin{equation*}
 \sum_{\xi\in \Lambda_{G_F,\alpha_F}}H_{M_{\lambda_1},N_{\chi_1}}^{F_{\xi}}\qdim_{V^G}F_{\xi}
  =\frac{\dim(W_{\chi_1})\dim(U_{\lambda_1})}{|G_M||G_N|}\frac{\qdim_V(F)}{|G_F|}\sum_{g,h,l\in G}N^{F\circ h}_{M\circ l,N\circ g}.
\end{equation*}
By Corollary \ref{coro:qdim_partial}, we have
\begin{equation*}
 \sum_{\xi\in \Lambda_{G_F,\alpha_F}}N_{M_{\lambda_1},N_{\chi_1}}^{F_{\xi}}\qdim_{V^G}F_{\xi}
 =\sum_{\xi\in \Lambda_{G_F,\alpha_F}}H_{M_{\lambda_1},N_{\chi_1}}^{F_{\xi}}\qdim_{V^G}F_{\xi}.
\end{equation*}
Since we have $\qdim_{V^G}(F_\xi)>0$ and \(N_{M_{\lambda_1},N_{\chi_1}}^{F_{\xi}}\geq H_{M_{\lambda_1},N_{\chi_1}}^{F_{\xi}}\), the desired result follows.
\end{proof}
\bibliographystyle{plain}
\bibliography{LiuyiZhangRef}
\end{document}